\documentclass[11pt]{article}
\usepackage[a4paper]{geometry}
\usepackage{amsfonts, amsmath, amssymb, amsthm, graphicx, caption, authblk, multirow, makecell, framed, float, xcolor, enumitem, tikz, hyperref}
\usetikzlibrary{fit,shapes.arrows}
\theoremstyle{definition}

\newtheorem{lemma}{Lemma}

\theoremstyle{remark}

\newcommand*{\mybox}[1]{%
  \framebox{\raisebox{0cm}[0.5\baselineskip][0.05\baselineskip]{%
    \hbox to 0.10cm {\hss#1\hss}}}\hspace{0.05cm}}

\begin{document}
\title{Verifying the First Nonzero Term: Physical ZKPs for ABC End View, Goishi Hiroi, and Toichika}
\author[1]{Suthee Ruangwises\thanks{\texttt{ruangwises@uec.ac.jp}}}
\affil[1]{Department of Informatics, The University of Electro-Communications, Tokyo, Japan}
\date{}
\maketitle

\begin{abstract}
In this paper, we propose a physical protocol to verify the first nonzero term of a sequence using a deck of cards. The protocol lets a prover show the value of the first nonzero term of a given sequence to a verifier without revealing which term it is. Our protocol uses $\Theta(1)$ shuffles, which is asymptotically lower than that of an existing protocol of Fukusawa and Manabe which uses $\Theta(n)$ shuffles, where $n$ is the length of the sequence. We also apply our protocol to construct zero-knowledge proof protocols for three well-known logic puzzles: ABC End View, Goishi Hiroi, and Toichika. These protocols enables a prover to physically show that he/she know solutions of the puzzles without revealing them.

\end{abstract}

\section{Introduction}
First introduced by Goldwasser et al. \cite{zkp0} in 1989, a \textit{zero-knowledge proof (ZKP)} is an interactive protocol which lets a prover $P$ convince a verifier $V$ that some statement is correct without revealing any other information. A ZKP has to satisfy the following properties.
\begin{enumerate}
	\item \textbf{Completeness:} If the statement is true, then $V$ accepts with high probability. (In this paper, we consider only the \textit{perfect completeness} property where $V$ always accepts.)
	\item \textbf{Soundness:} If the statement is false, then $V$ rejects with high probability. (In this paper, we consider only the \textit{perfect soundness} property where $V$ always rejects.)
	\item \textbf{Zero-knowledge:} During the verification, $V$ gains no extra information other than the correctness of the statement.
\end{enumerate}

It has been proved that for any NP problem, there exists a computational ZKP for it \cite{zkp}. As a result, one can construct a computational ZKP for any NP problem via a reduction. Such construction, however, is not practical or intuitive. Instead, many researchers focus on using physical objects such as a deck of cards to develop ZKP protocols. These card-based protocols have benefits that they use only portable objects found in everyday life and do not require computers. These protocols also allow external observers to verify that the prover truthfully executes them (which is a difficult task for digital protocols). Furthermore, these protocol are easy to understand, even for non-expert, and thus can be used as examples for didactic purpose.

\subsection{Related Work}
Card-based ZKP protocols for several logic puzzles have been developed, including Sudoku \cite{sudoku0,sudoku2,sudoku}, Nonogram \cite{nonogram,nonogram2}, Akari \cite{akari}, Takuzu \cite{akari,takuzu}, Kakuro \cite{akari,kakuro}, KenKen \cite{akari}, Makaro \cite{makaro,makaro2}, Norinori \cite{norinori}, Slitherlink \cite{slitherlink}, Juosan \cite{takuzu}, Numberlink \cite{numberlink}, Suguru \cite{suguru}, Ripple Effect \cite{ripple}, Nurikabe \cite{nurikabe}, Hitori \cite{nurikabe}, Bridges \cite{bridges}, Masyu \cite{slitherlink}, Heyawake \cite{nurikabe}, Shikaku \cite{shikaku}, Usowan \cite{usowan}, Nurimisaki \cite{nurimisaki}, ABC End View \cite{abc,goishi}, Ball sort puzzle \cite{ball}, Moon-or-Sun \cite{moon}, Kurodoko \cite{nurimisaki}, Five Cells \cite{fivecells}, and Sumplete \cite{sumplete}.

These protocols are not only for recreational purpose, but they also have theoretical importance; they can physically verify many important sequence-related functions, including the ones shown below.
\begin{itemize}
	\item A subprotocol in \cite{makaro} can verify that a term of a sequence is the largest one in that sequence without revealing any term of the sequence.
	\item A subprotocol in \cite{sudoku} can verify that a sequence is a permutation of all given numbers without revealing their order.
	\item A subprotocol in \cite{numberlink} can count the number of terms of a sequence that are equal to a given secret value without revealing that secret value, which terms are equal to it, or any term of the sequence.
	\item A subprotocol in \cite{ripple}, when given a secret number $x$ and a sequence, can verify that $x$ does not appear among the first $x$ terms of the sequence without revealing $x$ or any term of the sequence.
\end{itemize}

\subsubsection{Protocol of Fukusawa and Manabe}
One challenging work is to develop a protocol to verify the first nonzero term of a sequence. Suppose a sequence $S=(x_1,x_2,...,x_n)$ of $n$ integers is known only to $P$. $P$ wants to show $V$ the value of the first nonzero term of $S$ without revealing which term it is, i.e. show the value of $x_i$ with the smallest index $i$ such that $x_i \neq 0$ without revealing $i$. Optionally, $P$ may also publicly edit the value of $x_i$.

In 2022, Fukasawa and Manabe \cite{abc} proposed such protocol in the context of developing a ZKP protocol for a pencil-and-paper logic puzzle called \textit{ABC End View}. However, their protocol involves computing multiple Boolean functions, making it very complicated. Another drawback of their protocol is that it encodes each term of the sequence by three cards (one numbered card and two binary cards: \mybox{$\clubsuit$}, \mybox{$\heartsuit$}), which is not optimal. The protocol also uses as many as $\Theta(n)$ shuffles, where $n$ is the length of the sequence, making it impractical to implement in real world.

\subsection{Our Contribution}
In this paper, we develop a simpler protocol to verify the first nonzero term of a sequence. Our protocol uses only $\Theta(1)$ shuffles and uses only one card to encode each term.

We apply our protocol to develop ZKP protocols for three well-known logic puzzles: ABC End View (the same as in \cite{abc} but with asymptotically better performance), \textit{Goishi Hiroi}, and \textit{Toichika}.

The major difference from the conference version of this paper \cite{goishi} is the addition of the ZKP protocol for Toichika and its proof of correctness and security in Section \ref{tck}.

\section{Preliminaries}
\subsection{Cards}
Each card used in our protocol has an integer written on the front side. All cards have indistinguishable back sides denoted by \mybox{?}.

For $1 \leq x \leq q$, define $E_q(x)$ to be a sequence of $q$ cards, with all of them being \mybox{0}s except the $x$-th leftmost card being a \mybox{1}, e.g. $E_4(2)$ is \mbox{\mybox{0}\mybox{1}\mybox{0}\mybox{0}}. Also, define $E_q(0)$ to be a sequence of $q$ consecutive \mbox{\mybox{0}s} and $E_q(q+1)$ to be a sequence of $q$ consecutive \mbox{\mybox{1}s}. In some operations, we may stack the cards in $E_q(x)$ into a single stack (with the leftmost card being the topmost card of the stack).

\subsection{Pile-Shifting Shuffle}
Given a matrix $M$ of cards, a \textit{pile-shifting shuffle} \cite{polygon} shifts the columns of $M$ by a uniformly random cyclic shift unknown to all parties (see Fig. \ref{fig1}). It can be implemented in real world by putting all cards in each column into an envelope, and letting all parties take turns to apply \textit{Hindu cuts} (taking some envelopes from the bottom and putting them on the top) to the pile of envelopes \cite{hindu}.

Note that each card in $M$ can be replaced by a stack of cards, and the protocol still works in the same way as long as every stack in the same row consists of the same number of cards.

\begin{figure}[H]
\centering
\begin{tikzpicture}
\node at (0,0.6) {\mybox{?}};
\node at (0.5,0.6) {\mybox{?}};
\node at (1,0.6) {\mybox{?}};
\node at (1.5,0.6) {\mybox{?}};
\node at (2,0.6) {\mybox{?}};

\node at (0,1.2) {\mybox{?}};
\node at (0.5,1.2) {\mybox{?}};
\node at (1,1.2) {\mybox{?}};
\node at (1.5,1.2) {\mybox{?}};
\node at (2,1.2) {\mybox{?}};

\node at (0,1.8) {\mybox{?}};
\node at (0.5,1.8) {\mybox{?}};
\node at (1,1.8) {\mybox{?}};
\node at (1.5,1.8) {\mybox{?}};
\node at (2,1.8) {\mybox{?}};

\node at (0,2.4) {\mybox{?}};
\node at (0.5,2.4) {\mybox{?}};
\node at (1,2.4) {\mybox{?}};
\node at (1.5,2.4) {\mybox{?}};
\node at (2,2.4) {\mybox{?}};

\node at (-0.4,0.6) {4};
\node at (-0.4,1.2) {3};
\node at (-0.4,1.8) {2};
\node at (-0.4,2.4) {1};

\node at (0,2.9) {1};
\node at (0.5,2.9) {2};
\node at (1,2.9) {3};
\node at (1.5,2.9) {4};
\node at (2,2.9) {5};

\node at (2.9,1.5) {\LARGE{$\Rightarrow$}};
\end{tikzpicture}
\begin{tikzpicture}
\node at (0,0.6) {\mybox{?}};
\node at (0.5,0.6) {\mybox{?}};
\node at (1,0.6) {\mybox{?}};
\node at (1.5,0.6) {\mybox{?}};
\node at (2,0.6) {\mybox{?}};

\node at (0,1.2) {\mybox{?}};
\node at (0.5,1.2) {\mybox{?}};
\node at (1,1.2) {\mybox{?}};
\node at (1.5,1.2) {\mybox{?}};
\node at (2,1.2) {\mybox{?}};

\node at (0,1.8) {\mybox{?}};
\node at (0.5,1.8) {\mybox{?}};
\node at (1,1.8) {\mybox{?}};
\node at (1.5,1.8) {\mybox{?}};
\node at (2,1.8) {\mybox{?}};

\node at (0,2.4) {\mybox{?}};
\node at (0.5,2.4) {\mybox{?}};
\node at (1,2.4) {\mybox{?}};
\node at (1.5,2.4) {\mybox{?}};
\node at (2,2.4) {\mybox{?}};

\node at (-0.4,0.6) {4};
\node at (-0.4,1.2) {3};
\node at (-0.4,1.8) {2};
\node at (-0.4,2.4) {1};

\node at (0,2.9) {3};
\node at (0.5,2.9) {4};
\node at (1,2.9) {5};
\node at (1.5,2.9) {1};
\node at (2,2.9) {2};
\end{tikzpicture}
\caption{An example of a pile-shifting shuffle on a $4 \times 5$ matrix}
\label{fig1}
\end{figure}

\subsection{Chosen Cut Protocol} \label{chosen}
Given a sequence $C = (c_1,c_2,...,c_q)$ of $q$ face-down cards, a \textit{chosen cut protocol} \cite{koch} lets $P$ select a desired card $c_i$ to use in other protocols without revealing $i$ to $V$. The protocol also preserves the original state of $C$ after $P$ finishes using $c_i$.

\begin{figure}[H]
\centering
\begin{tikzpicture}
\node at (0.0,2.4) {\mybox{?}};
\node at (0.6,2.4) {\mybox{?}};
\node at (1.2,2.4) {...};
\node at (1.8,2.4) {\mybox{?}};
\node at (2.4,2.4) {\mybox{?}};
\node at (3.0,2.4) {\mybox{?}};
\node at (3.6,2.4) {...};
\node at (4.2,2.4) {\mybox{?}};

\node at (0.0,2) {$c_1$};
\node at (0.6,2) {$c_2$};
\node at (1.8,2) {$c_{i-1}$};
\node at (2.4,2) {$c_i$};
\node at (3.0,2) {$c_{i+1}$};
\node at (4.2,2) {$c_q$};

\node at (0.0,1.4) {\mybox{?}};
\node at (0.6,1.4) {\mybox{?}};
\node at (1.2,1.4) {...};
\node at (1.8,1.4) {\mybox{?}};
\node at (2.4,1.4) {\mybox{?}};
\node at (3.0,1.4) {\mybox{?}};
\node at (3.6,1.4) {...};
\node at (4.2,1.4) {\mybox{?}};

\node at (0.0,1) {0};
\node at (0.6,1) {0};
\node at (1.8,1) {0};
\node at (2.4,1) {1};
\node at (3.0,1) {0};
\node at (4.2,1) {0};

\node at (0.0,0.4) {\mybox{1}};
\node at (0.6,0.4) {\mybox{0}};
\node at (1.2,0.4) {...};
\node at (1.8,0.4) {\mybox{0}};
\node at (2.4,0.4) {\mybox{0}};
\node at (3.0,0.4) {\mybox{0}};
\node at (3.6,0.4) {...};
\node at (4.2,0.4) {\mybox{0}};
\end{tikzpicture}
\caption{A $3 \times n$ matrix $M$ constructed in Step 1 of the chosen cut protocol}
\label{fig2}
\end{figure}

\begin{enumerate}
	\item Construct the following $3 \times q$ matrix $M$ (see Fig. \ref{fig2}).
	\begin{enumerate}
		\item In Row 1, place the sequence $C$.
		\item In Row 2, place a face-down sequence $E_q(i)$.
		\item In Row 3, place a face-up sequence $E_q(1)$.
	\end{enumerate}
	\item Turn over all face-up cards. Apply the pile-shifting shuffle to $M$.
	\item Turn over all cards in Row 2 and locate the position of the only \mybox{1}. A card in Row 1 directly above this \mybox{1} will be the card $c_i$ as desired.
	\item After finishing using $c_i$ in other protocols, place $c_i$ back into $M$ at the same position.
	\item Turn over all face-up cards. Apply the pile-shifting shuffle to $M$.
	\item Turn over all cards in Row 3 and locate the position of the only \mybox{1}. Shift the columns of $M$ cyclically such that this \mybox{1} moves to Column 1. $M$ is now reverted back to its original state.
\end{enumerate}

Note that each card in $C$ can be replaced by a stack of cards, and the protocol still works in the same way as long as every stack consists of the same number of cards.

\section{Verifying the First Nonzero Term of a Sequence} \label{nonzero}
Suppose a sequence $S=(x_1,x_2,...,x_n)$ of $n$ integers is known only to $P$. $P$ wants to show $V$ the value of the first nonzero term of $S$ without revealing which term it is, i.e. show the value of $x_i$ with the smallest index $i$ such that $x_i \neq 0$ without revealing $i$. Optionally, $P$ may also publicly edit the value of $x_i$. We develop the following \textit{FirstNonZero} protocol, which can solve this problem.

Suppose $x_\ell$ is the first nonzero term of $S$. The sequence $S$ is encoded by a sequence $A=(a_1,a_2,...,$ $a_n)$ of face-down cards, with each $a_i$ being an \mybox{$x_i$}. $P$ performs the following steps.
\begin{enumerate}
	\item Publicly append cards $b_1,b_2,...,b_{n-1}$, all of them being $\mybox{0}$s, to the left of $A$. Call this new sequence $C=(c_1,c_2,...,c_{2n-1}) = (b_1,b_2,...,b_{n-1},a_1,a_2,...,a_n)$.
	\item Turn over all face-up cards. Apply the chosen cut protocol to $C$ to choose the card $c_{\ell+n-1}=a_\ell$.
	\item Turn over $c_{\ell+n-1}$ to show that it is not a \mybox{0} (otherwise $V$ rejects). Also, as the chosen cut protocol preserves the cyclic order of $C$, turn over cards $c_\ell,c_{\ell+1},...,c_{\ell+n-2}$ to show that they are all \mybox{0}s (otherwise $V$ rejects). Now $V$ is convinced that the shown number on $c_{\ell+n-1}$ is the first nonzero term of $S$, i.e. $x_\ell \neq 0$ and $x_1=x_2=...=x_{\ell-1}=0$, without knowing $\ell$.
	\item Optionally, $P$ may publicly replace the card $c_{\ell+n-1}$ with another desired card.
	\item Continue the chosen cut protocol to the end. Then, remove cards $b_1,b_2,...,$ $b_{n-1}$. The sequence $A$ is now reverted back to its original state.
\end{enumerate}

This protocol uses $\Theta(n)$ cards and $\Theta(1)$ shuffles, in contrast to the protocol of Fukusawa and Manabe \cite{abc}, which uses $\Theta(n)$ cards and $\Theta(n)$ shuffles. Also, our protocol uses only a single card to encode each number, while their protocol uses three. This will affect the long-term performance when the protocol is repeatedly applied.

\section{ZKP Protocol for ABC End View}
ABC End View (also known as \textit{Easy as ABC}) is a pencil-and-paper logic puzzle. The origin of this puzzle is unclear, but it is a standard puzzle used in many competitions and appears on many websites. The puzzle consists of an $n \times n$ initially empty grid, with some letters written outside the edge of the grid. The objective of the puzzle is to fill letters from a given range (e.g. A, B, and C) into some cells in the grid according to the following rules (see Fig. \ref{fig3}).
\begin{enumerate}
	\item Every row and column must contain each letter exactly once.
	\item A letter outside the edge of the grid indicates the first letter in the corresponding row or column from that direction.
\end{enumerate}

\begin{figure}
\centering
\begin{tikzpicture}
\draw[step=0.8cm] (0,0) grid (4,4);

\node at (2.8,-0.2) {A};
\node at (3.6,-0.2) {A};
\node at (4.2,0.4) {B};
\node at (-0.2,1.2) {B};
\node at (-0.2,2.8) {A};
\node at (4.2,2.8) {C};
\node at (0.4,4.2) {B};
\node at (1.2,4.2) {C};
\end{tikzpicture}
\hspace{1.5cm}
\begin{tikzpicture}
\draw[step=0.8cm,color=black] (0,0) grid (4,4);

\node at (2.8,-0.2) {A};
\node at (3.6,-0.2) {A};
\node at (4.2,0.4) {B};
\node at (-0.2,1.2) {B};
\node at (-0.2,2.8) {A};
\node at (4.2,2.8) {C};
\node at (0.4,4.2) {B};
\node at (1.2,4.2) {C};

\node at (0.4,0.4) {C};
\node at (1.2,0.4) {A};
\node at (2.0,0.4) {B};
\node at (1.2,1.2) {B};
\node at (2.0,1.2) {C};
\node at (3.6,1.2) {A};
\node at (1.2,2.0) {C};
\node at (2.8,2.0) {A};
\node at (3.6,2.0) {B};
\node at (0.4,2.8) {A};
\node at (2.8,2.8) {B};
\node at (3.6,2.8) {C};
\node at (0.4,3.6) {B};
\node at (2.0,3.6) {A};
\node at (2.8,3.6) {C};
\end{tikzpicture}
\caption{An example of a $5 \times 5$ ABC End View puzzle with letters from the range A, B, and C (left) and its solution (right)}
\label{fig3}
\end{figure}

The construction of a ZKP protocol for ABC End View is very straightforward. We encode a letter A by a \mybox{1}, a letter B by a \mybox{2}, a letter C by a \mybox{3}, and so on. Also, we encode an empty cell by a \mybox{0}. Hence, we can directly apply the FirstNonZero protocol to verify the second rule for each letter outside the grid. To verify the first rule, we apply the following \textit{uniqueness verification protocol} for each row and column.

\subsection{Uniqueness Verification Protocol} \label{unique}
The uniqueness verification protocol \cite{sudoku} allows $P$ to convince $V$ that a sequence $\sigma$ of $q$ face-down cards is a permutation of cards $a_1,a_2,...,a_q$ in some order, without revealing the order. The protocol also preserves the original state of $\sigma$.

$P$ performs the following steps.

\begin{figure}[H]
\centering
\begin{tikzpicture}
\node at (-0.5,1.2) {$\sigma$:};
\node at (0.0,1.2) {\mybox{?}};
\node at (0.6,1.2) {\mybox{?}};
\node at (1.2,1.2) {...};
\node at (1.8,1.2) {\mybox{?}};

\node at (0.0,0.4) {\mybox{1}};
\node at (0.6,0.4) {\mybox{2}};
\node at (1.2,0.4) {...};
\node at (1.8,0.4) {\mybox{$q$}};
\end{tikzpicture}
\caption{A $2 \times q$ matrix constructed in Step 1}
\label{fig33}
\end{figure}

\begin{enumerate}
	\item Place face-up cards \mybox{1}, \mybox{2}, ..., \mybox{$q$} below the sequence $\sigma$ in this order from left to right to form a $2 \times q$ matrix of cards (see Fig \ref{fig33}).
	\item Turn over all face-up cards. Rearrange all columns of the matrix by a uniformly random permutation (which can be implemented in real world by putting the two cards in each column into an envelope and scrambling all envelopes together).
	\item Turn over all cards in Row 1 to show that they are a permutation of $a_1,a_2,...,a_q$ (otherwise, $V$ rejects).
	\item Turn over all face-up cards. Rearrange all columns of the matrix by a uniformly random permutation.
	\item Turn over all cards in Row 2. Rearrange the columns such that the cards in the bottom rows are \mybox{1}, \mybox{2}, ..., \mybox{$q$} in this order from left to right. The sequence $\sigma$ is now reverted back to its original state.
\end{enumerate}

Our ZKP protocol for ABC End View uses $\Theta(n^2)$ cards and $\Theta(n)$ shuffles, in contrast to the protocol of Fukusawa and Manabe \cite{abc}, which uses $\Theta(n^2)$ cards and $\Theta(n^2)$ shuffles. Also, as our protocol uses one card to encode each cell, the actual number of required cards is $n^2+\Theta(n)$, while their protocol uses three cards to encode each cell, resulting in the total of $3n^2+\Theta(n)$ cards.

The proof of correctness and security of this protocol is omitted as it is just a straightforward application of the FirstNonZero protocol and the uniqueness verification protocol, plus a similar proof was already provided in \cite{abc}.

\section{ZKP Protocol for Goishi Hiroi}
Goishi Hiroi (also known as \textit{Hiroimono}) is a logic puzzle first published by a Japanese company Nikoli, which is famous for publishing many popular logic puzzles such as Sudoku, Kakuro, and Akari. The puzzle consists of $m$ stones placed on an $n \times n$ grid, with at most one stone in each cell. The objective of the puzzle is to pick all stones in order according to the following rules \cite{nikoli}.
\begin{enumerate}
	\item The player can start picking at any stone.
	\item After picking a stone, the player must travel horizontally or vertically to the next stone.
	\item During the travel, if there is a stone on the path, the player must pick it. After that, the player may continue traveling in the same direction or turn left or right, but cannot turn backward (see Fig. \ref{fig4}).
\end{enumerate}

\begin{figure}
\centering
\begin{tikzpicture}
\draw[step=0.8cm] (0,0) grid (4.8,4.8);

\node[draw,circle,minimum size=0.7cm,fill={rgb:black,1;white,7}] at (0.4,0.4) {};
\node[draw,circle,minimum size=0.7cm,fill={rgb:black,1;white,7}] at (0.4,1.2) {};
\node[draw,circle,minimum size=0.7cm,fill={rgb:black,1;white,7}] at (0.4,4.4) {};
\node[draw,circle,minimum size=0.7cm,fill={rgb:black,1;white,7}] at (1.2,1.2) {};
\node[draw,circle,minimum size=0.7cm,fill={rgb:black,1;white,7}] at (1.2,2) {};
\node[draw,circle,minimum size=0.7cm,fill={rgb:black,1;white,7}] at (2.8,2.8) {};
\node[draw,circle,minimum size=0.7cm,fill={rgb:black,1;white,7}] at (2.8,3.6) {};
\node[draw,circle,minimum size=0.7cm,fill={rgb:black,1;white,7}] at (3.6,2.8) {};
\node[draw,circle,minimum size=0.7cm,fill={rgb:black,1;white,7}] at (4.4,0.4) {};
\node[draw,circle,minimum size=0.7cm,fill={rgb:black,1;white,7}] at (4.4,1.2) {};
\node[draw,circle,minimum size=0.7cm,fill={rgb:black,1;white,7}] at (4.4,2.8) {};
\node[draw,circle,minimum size=0.7cm,fill={rgb:black,1;white,7}] at (4.4,3.6) {};
\end{tikzpicture}
\hspace{1.2cm}
\begin{tikzpicture}
\draw[step=0.8cm] (0,0) grid (4.8,4.8);

\node[draw,circle,minimum size=0.7cm,fill={rgb:black,1;white,7}] at (0.4,0.4) {10};
\node[draw,circle,minimum size=0.7cm,fill={rgb:black,1;white,7}] at (0.4,1.2) {11};
\node[draw,circle,minimum size=0.7cm,fill={rgb:black,1;white,7}] at (0.4,4.4) {12};
\node[draw,circle,minimum size=0.7cm,fill={rgb:black,1;white,7}] at (1.2,1.2) {2};
\node[draw,circle,minimum size=0.7cm,fill={rgb:black,1;white,7}] at (1.2,2) {1};
\node[draw,circle,minimum size=0.7cm,fill={rgb:black,1;white,7}] at (2.8,2.8) {6};
\node[draw,circle,minimum size=0.7cm,fill={rgb:black,1;white,7}] at (2.8,3.6) {7};
\node[draw,circle,minimum size=0.7cm,fill={rgb:black,1;white,7}] at (3.6,2.8) {5};
\node[draw,circle,minimum size=0.7cm,fill={rgb:black,1;white,7}] at (4.4,0.4) {9};
\node[draw,circle,minimum size=0.7cm,fill={rgb:black,1;white,7}] at (4.4,1.2) {3};
\node[draw,circle,minimum size=0.7cm,fill={rgb:black,1;white,7}] at (4.4,2.8) {4};
\node[draw,circle,minimum size=0.7cm,fill={rgb:black,1;white,7}] at (4.4,3.6) {8};
\end{tikzpicture}
\caption{An example of a $6 \times 6$ Goishi Hiroi puzzle with 12 stones (left) and its solution with each number $i$ indicating the $i$-th stone that was picked (right)}
\label{fig4}
\end{figure}

The problem of deciding if a given Goishi Hiroi instance has a solution is known to be NP-complete \cite{np}.

We will apply the FirstNonZero protocol to develop a ZKP protocol for Goishi Hiroi.

\subsection{Idea of the Protocol}
First, we apply the chosen cut protocol to select a starting stone. Then, we take the $n-1$ stones on each path from the selected stone in the direction to the north, east, south, and west (we extend the grid by $n-1$ rows and $n-1$ columns to support this), and apply the chosen cut protocol again to select the travel direction. After that, we apply the FirstNonZero protocol to select the first stone on that path.

Note that we have to keep track of the direction we are currently traveling; its opposite direction will be the ``forbidden direction'' we cannot travel in the next move. Therefore, in each move, we also have to verify that the selected direction is not the forbidden direction.

\subsection{Setup}
$P$ publicly places a \mybox{1} on each cell with a stone and a \mybox{0} on each empty cell in the Goishi Hiroi grid. $P$ also publicly appends $n-1$ rows and $n-1$ columns of ``dummy cards'' $\mybox{3}$s to the bottom and to the right of the grid. Then, turn over all face-up cards. We now have a $(2n-1) \times (2n-1)$ matrix of cards (see Fig. \ref{fig44}).

\begin{figure}
\centering
\begin{tikzpicture}
\draw[step=0.8cm] (0,0) grid (2.4,2.4);

\node[draw,circle,minimum size=0.7cm,fill={rgb:black,1;white,7}] at (0.4,0.4) {};
\node[draw,circle,minimum size=0.7cm,fill={rgb:black,1;white,7}] at (0.4,1.2) {};
\node[draw,circle,minimum size=0.7cm,fill={rgb:black,1;white,7}] at (2,1.2) {};
\node[draw,circle,minimum size=0.7cm,fill={rgb:black,1;white,7}] at (2,2) {};

\node at (2,-0.2) {};
\end{tikzpicture}
\hspace{0.2cm}
\begin{tikzpicture}
\node at (0.0,1.2) {\mybox{0}};
\node at (0.6,1.2) {\mybox{0}};
\node at (1.2,1.2) {\mybox{1}};

\node at (0.0,0.6) {\mybox{1}};
\node at (0.6,0.6) {\mybox{0}};
\node at (1.2,0.6) {\mybox{1}};

\node at (0.0,0.0) {\mybox{1}};
\node at (0.6,0.0) {\mybox{0}};
\node at (1.2,0.0) {\mybox{0}};

\node at (3.0,-0.8) {};

\node at (2.2,0.6) {\LARGE{$\Rightarrow$}};
\node at (-1,0.6) {\LARGE{$\Rightarrow$}};
\end{tikzpicture}
\begin{tikzpicture}
\node at (1.2,2.4) {\mybox{0}};
\node at (1.8,2.4) {\mybox{0}};
\node at (2.4,2.4) {\mybox{1}};
\node at (3.0,2.4) {\mybox{3}};
\node at (3.6,2.4) {\mybox{3}};

\node at (1.2,1.8) {\mybox{1}};
\node at (1.8,1.8) {\mybox{0}};
\node at (2.4,1.8) {\mybox{1}};
\node at (3.0,1.8) {\mybox{3}};
\node at (3.6,1.8) {\mybox{3}};

\node at (1.2,1.2) {\mybox{1}};
\node at (1.8,1.2) {\mybox{0}};
\node at (2.4,1.2) {\mybox{0}};
\node at (3.0,1.2) {\mybox{3}};
\node at (3.6,1.2) {\mybox{3}};

\node at (1.2,0.6) {\mybox{3}};
\node at (1.8,0.6) {\mybox{3}};
\node at (2.4,0.6) {\mybox{3}};
\node at (3.0,0.6) {\mybox{3}};
\node at (3.6,0.6) {\mybox{3}};

\node at (1.2,0) {\mybox{3}};
\node at (1.8,0) {\mybox{3}};
\node at (2.4,0) {\mybox{3}};
\node at (3.0,0) {\mybox{3}};
\node at (3.6,0) {\mybox{3}};
\end{tikzpicture}
\caption{The way we place cards on a $3 \times 3$ Goishi Hiroi grid during the setup}
\label{fig44}
\end{figure}

Note that if we arrange all cards in the matrix into a single sequence $C=(c_1,c_2,...,$ $c_{(2n-1)^2})$, starting at the top-left corner and going from left to right in Row 1, then from left to right in Row 2, and so on, we can locate exactly where the four neighbors of any given card are. Namely, the cards on the neighbor to the north, east, south, and west of a cell containing $c_i$ are $c_{i-2n+1}$, $c_{i+1}$, $c_{i+2n-1}$, and $c_{i-1}$, respectively.

\subsection{Main Protocol}
To pick the first stone, $P$ performs the following steps.
\begin{enumerate}
	\item Apply the chosen cut protocol to select a card corresponding to the first stone.
	\item Turn over the selected card to show that it is a \mybox{1} (otherwise $V$ rejects). Replace it with a \mybox{2} and place it back to the grid.
\end{enumerate}

To pick the second stone, $P$ performs the following steps.
\begin{enumerate}
	\item Apply the chosen cut protocol to select a card corresponding to the first stone.
	\item Turn over the selected card to show that it is a \mybox{2} (otherwise $V$ rejects). Replace it with a \mybox{0}.
	\item Take the $n-1$ cards on a path from the selected card in the direction to the north. Let $A_1=(a_{(1,1)},a_{(1,2)},...,a_{(1,n-1)})$ be the sequence of these cards in this order from the nearest to the farthest. Analogously, let $A_2, A_3$, and $A_4$ be the sequences of the $n-1$ cards on a path from the selected card in the direction to the east, south, and west, respectively. Stack each sequence into a single stack.
	\item Place a \mybox{0}, called $a_{(i,0)}$, on top of $A_i$ for each $i=1,2,3,4$. Now we have $A_i=(a_{(i,0)},a_{(i,1)},...,a_{(i,n-1)})$ for $i=1,2,3,4$.
	\item Apply the chosen cut protocol to select a stack $A_\ell$ corresponding to the direction towards the second stone.
	\item Apply the FirstNonZero protocol to the sequence $(a_{(\ell,1)},a_{(\ell,2)},...,a_{(\ell,n-1)})$ to select a card corresponding to the second stone. $V$ verifies that it is a \mybox{1} (otherwise $V$ rejects). Replace the selected card with a \mybox{2}.
	\item Replace $a_{(\ell,0)}$ with a \mybox{1}.
	\item Retain cards $a_{(1,0)},a_{(2,0)},a_{(3,0)},a_{(4,0)}$ and place the rest of the cards back to the grid.
\end{enumerate}

To pick each $p$-th stone for $p \geq 3$, the steps are very similar to picking the second stone. The only additional step is that, after $P$ selects a direction, $P$ has to show $V$ that it is not a forbidden direction. The formal steps are as follows.
\begin{enumerate}
	\item Consider the cards $a_{(1,0)},a_{(2,0)},a_{(3,0)},a_{(4,0)}$ retained from the previous iteration. Swap $a_{(1,0)}$ and $a_{(3,0)}$. Swap $a_{(2,0)}$ and $a_{(4,0)}$. (The forbidden direction in this iteration is the one opposite to the travel direction in the previous iteration.)
	\item Apply the chosen cut protocol to select a card corresponding to the $(p-1)$-th stone.
	\item Turn over the selected card to show that it is a \mybox{2} (otherwise $V$ rejects). Replace it with a \mybox{0}.
	\item Take the $n-1$ cards on a path from the selected card in the direction to the north. Let $A_1=(a_{(1,1)},a_{(1,2)},...,a_{(1,n-1)})$ be the sequence of these cards in this order from the nearest to the farthest. Analogously, let $A_2, A_3$, and $A_4$ be the sequences of the $n-1$ cards on a path from the selected card in the direction to the east, south, and west, respectively. Stack each sequence into a single stack.
	\item Place the card $a_{(i,0)}$ on top of $A_i$ for each $i=1,2,3,4$. Now we have $A_i=(a_{(i,0)},a_{(i,1)},$ $...,a_{(i,n-1)})$ for $i=1,2,3,4$.
	\item Apply the chosen cut protocol to select a stack $A_\ell$ corresponding to the direction towards the $p$-th stone.
	\item Turn over the card $a_{(\ell,0)}$ to show that it is a \mybox{0} (otherwise $V$ rejects).
	\item Apply the FirstNonZero protocol to the sequence $(a_{(\ell,1)},a_{(\ell,2)},...,a_{(\ell,n-1)})$ to select a card corresponding to the $p$-th stone. $V$ verifies that it is a \mybox{1} (otherwise $V$ rejects). Replace the selected card with a \mybox{2}.
	\item Replace $a_{(\ell,0)}$ with a \mybox{1}. Also, replace each $a_{(i,0)}$ for $i \neq \ell$ with a \mybox{0}.
	\item Retain cards $a_{(1,0)},a_{(2,0)},a_{(3,0)},a_{(4,0)}$ and place the rest of the cards back to the grid.
\end{enumerate}

If the verification passes for all $p=3,4,...,m$, then $V$ accepts.

Our ZKP protocol for Goishi Hiroi uses $\Theta(n^2)$ cards and $\Theta(m)$ shuffles.

\subsection{Proof of Correctness and Security}
We will prove the perfect completeness, perfect soundness, and zero-knowledge properties of our protocol for Goishi Hiroi.

\begin{lemma}[Perfect Completeness] \label{lem1}
If $P$ knows a solution of the Goishi Hiroi puzzle, then $V$ always accepts.
\end{lemma}

\begin{proof}
Suppose $P$ knows a solution of the puzzle. Consider when $P$ picks the $p$-th stone ($p \geq 3$) from the grid.

At the beginning of Step 1, the only card $a_{(i,0)}$ that is a \mybox{1} is the one corresponding to the travel direction in the previous iteration. So, after $P$ swaps the cards, the only \mybox{1} will be the one corresponding to the opposite direction of the travel direction in the previous iteration.

In Step 3, the selected card was changed to \mybox{2} in the previous iteration, so this step will pass.

In Step 7, since the travel direction cannot be opposite to that of the previous iteration, the card must be a \mybox{0}, so this step will pass.

In Step 8, since the stone has not been picked before, the card must be a \mbox{\mybox{1},} so the verification will pass. Also, when invoking the FirstNonZero protocol, since the stone is the first one on the path, the protocol will also pass.

As this is true for every $p \geq 3$, and the case $p=2$ also works similarly, while the case $p=1$ is trivial, we can conclude that $V$ always accepts.
\end{proof}

\begin{lemma}[Perfect Soundness] \label{lem2}
If $P$ does not know a solution of the Goishi Hiroi puzzle, then $V$ always rejects.
\end{lemma}

\begin{proof}
We will prove the contrapositive of the statement. Suppose that $V$ accepts, which means the verification passes for every iteration. Consider the $p$-th iteration ($p \geq 3$).

Since Step 3 passes, the card must be a \mybox{2}. As there is only one \mybox{2} in the grid, which is the card selected in the previous iteration, the move in this iteration must start from that cell.

Since Step 7 passes, the card must be a \mybox{0}, meaning that the current travel direction is not opposite to that of the previous iteration, satisfying the rule of the puzzle.

Since Step 8 passes, the card must be a \mybox{1}, meaning that there is a stone on the corresponding cell. Also, as the invoked FirstNonZero protocol also passes, the stone must be the first one on the path, satisfying the rule of the puzzle.

This means the $p$-th iteration corresponds to a valid move of picking a stone from the grid. As this is true for every $p \geq 3$, and the case $p=2$ also works similarly, while the case $p=1$ is trivial, we can conclude that $P$ must know a valid solution of the puzzle.
\end{proof}

\begin{lemma}[Zero-Knowledge] \label{lem3}
During the verification, $V$ obtains no information about $P$'s solution.
\end{lemma}

\begin{proof}
It is sufficient to show that all distributions of cards that are turned face-up can be simulated by a simulator $S$ that does not know $P$'s solution.
\begin{itemize}
	\item In Steps 3 and 6 of the chosen cut protocol in Section \ref{chosen}, due to the pile-shifting shuffle, a \mybox{1} has an equal probability to be at any of the $q$ positions. Hence, these steps can be simulated by $S$.
	\item In Step 3 of the FirstNonZero protocol in Section \ref{nonzero}, the cards $c_\ell,c_{\ell+1},...,$ $c_{\ell+n-2}$ are all \mybox{0}s, and the card $c_{\ell+n-1}$ is public information known to $V$. Hence, this step can be simulated by $S$.
	\item In the main protocol, there is only one deterministic pattern of the cards that are turned face-up, so the whole protocol can be simulated by $S$.
\end{itemize}

Therefore, we can conclude that $V$ obtains no information about $P$'s solution.
\end{proof}

\section{ZKP Protocol for Toichika} \label{tck}
Toichika is another pencil-and-paper logic puzzle developed by Nikoli. The puzzle consists of an $n \times n$ grid divided into polyominoes called \textit{regions}, with some cells already containing an arrow. The objective of the puzzle is to put arrows pointing in horizontal or vertical direction into cells according to the following rules (see Fig. \ref{fig5}).
\begin{enumerate}
	\item Each region contains exactly one arrow.
	\item Two arrows pointing towards each other with no other arrow between them are matched together; all arrows must be matched.
	\item The matched arrows cannot be in horizontally or vertically adjacent regions.
\end{enumerate}

\begin{figure}
\centering
\begin{tikzpicture}
\draw[step=0.8cm,color={rgb:black,1;white,4}] (0,0) grid (4.8,4.8);

\draw[line width=0.6mm] (0,0) -- (0,4.8);
\draw[line width=0.6mm] (4.8,0) -- (4.8,4.8);
\draw[line width=0.6mm] (0,0) -- (4.8,0);
\draw[line width=0.6mm] (0,4.8) -- (4.8,4.8);

\draw[line width=0.6mm] (0,0.8) -- (4,0.8);
\draw[line width=0.6mm] (0,1.6) -- (3.2,1.6);
\draw[line width=0.6mm] (4,1.6) -- (4.8,1.6);
\draw[line width=0.6mm] (0,2.4) -- (0.8,2.4);
\draw[line width=0.6mm] (3.2,2.4) -- (4.8,2.4);
\draw[line width=0.6mm] (0.8,3.2) -- (1.6,3.2);
\draw[line width=0.6mm] (2.4,3.2) -- (4,3.2);
\draw[line width=0.6mm] (0.8,4) -- (2.4,4);
\draw[line width=0.6mm] (3.2,4) -- (4.8,4);
\draw[line width=0.6mm] (0.8,2.4) -- (0.8,3.2);
\draw[line width=0.6mm] (0.8,4) -- (0.8,4.8);
\draw[line width=0.6mm] (1.6,0.8) -- (1.6,4);
\draw[line width=0.6mm] (2.4,3.2) -- (2.4,4.8);
\draw[line width=0.6mm] (3.2,0.8) -- (3.2,2.4);
\draw[line width=0.6mm] (3.2,4) -- (3.2,4.8);
\draw[line width=0.6mm] (4,0) -- (4,3.2);

\node[single arrow, draw=black, single arrow head extend=3pt, minimum height=5mm, rotate=90] at (2,0.4) {};
\node[single arrow, draw=black, single arrow head extend=3pt, minimum height=5mm, rotate=180] at (3.6,2) {};
\node[single arrow, draw=black, single arrow head extend=3pt, minimum height=5mm, rotate=180] at (3.6,3.6) {};
\node[single arrow, draw=black, single arrow head extend=3pt, minimum height=5mm, rotate=270] at (4.4,4.4) {};
\end{tikzpicture}
\hspace{1.5cm}
\begin{tikzpicture}
\draw[step=0.8cm,color={rgb:black,1;white,4}] (0,0) grid (4.8,4.8);

\draw[line width=0.6mm] (0,0) -- (0,4.8);
\draw[line width=0.6mm] (4.8,0) -- (4.8,4.8);
\draw[line width=0.6mm] (0,0) -- (4.8,0);
\draw[line width=0.6mm] (0,4.8) -- (4.8,4.8);

\draw[line width=0.6mm] (0,0.8) -- (4,0.8);
\draw[line width=0.6mm] (0,1.6) -- (3.2,1.6);
\draw[line width=0.6mm] (4,1.6) -- (4.8,1.6);
\draw[line width=0.6mm] (0,2.4) -- (0.8,2.4);
\draw[line width=0.6mm] (3.2,2.4) -- (4.8,2.4);
\draw[line width=0.6mm] (0.8,3.2) -- (1.6,3.2);
\draw[line width=0.6mm] (2.4,3.2) -- (4,3.2);
\draw[line width=0.6mm] (0.8,4) -- (2.4,4);
\draw[line width=0.6mm] (3.2,4) -- (4.8,4);
\draw[line width=0.6mm] (0.8,2.4) -- (0.8,3.2);
\draw[line width=0.6mm] (0.8,4) -- (0.8,4.8);
\draw[line width=0.6mm] (1.6,0.8) -- (1.6,4);
\draw[line width=0.6mm] (2.4,3.2) -- (2.4,4.8);
\draw[line width=0.6mm] (3.2,0.8) -- (3.2,2.4);
\draw[line width=0.6mm] (3.2,4) -- (3.2,4.8);
\draw[line width=0.6mm] (4,0) -- (4,3.2);

\node[single arrow, draw=black, single arrow head extend=3pt, minimum height=5mm, rotate=90] at (2,0.4) {};
\node[single arrow, draw=black, single arrow head extend=3pt, minimum height=5mm, rotate=180] at (3.6,2) {};
\node[single arrow, draw=black, single arrow head extend=3pt, minimum height=5mm, rotate=180] at (3.6,3.6) {};
\node[single arrow, draw=black, single arrow head extend=3pt, minimum height=5mm, rotate=270] at (4.4,4.4) {};
\node[single arrow, draw=black, single arrow head extend=3pt, minimum height=5mm, rotate=90] at (1.2,1.2) {};
\node[single arrow, draw=black, single arrow head extend=3pt, minimum height=5mm, rotate=0] at (2.8,1.2) {};
\node[single arrow, draw=black, single arrow head extend=3pt, minimum height=5mm, rotate=180] at (4.4,1.2) {};
\node[single arrow, draw=black, single arrow head extend=3pt, minimum height=5mm, rotate=0] at (0.4,2) {};
\node[single arrow, draw=black, single arrow head extend=3pt, minimum height=5mm, rotate=90] at (4.4,2) {};
\node[single arrow, draw=black, single arrow head extend=3pt, minimum height=5mm, rotate=270] at (2,2.8) {};
\node[single arrow, draw=black, single arrow head extend=3pt, minimum height=5mm, rotate=0] at (0.4,3.6) {};
\node[single arrow, draw=black, single arrow head extend=3pt, minimum height=5mm, rotate=270] at (1.2,4.4) {};
\end{tikzpicture}
\caption{An example of a Toichika puzzle (left) and its solution (right)}
\label{fig5}
\end{figure}

Deciding a solvability of a given Toichika instance has recently been proved to be NP-complete. \cite{toichika}.

We will apply the FirstNonZero protocol to develop a ZKP protocol for Toichika.

\subsection{Idea of the Protocol}
We encode each cell by two sequences of cards. The first one denotes the direction of an arrow in it (if any). The second one denotes the region it belongs to. The full details of the encoding will be explained in the next section.

To verify the second rule, we do it in a similar way to the ZKP protocol for Goishi Hiroi. Starting at an arrow, we take the $n-1$ stones in the direction to the north, east, south, and west, and apply the chosen cut protocol to select the direction that the arrow points to (we have to encode each arrow by four cards to support this). Then, we apply the FirstNonZero protocol to select the first arrow on the path, which is the matched arrow, and show that the directions of these two arrows are opposite. Note that we have to slightly modify the FirstNonZero protocol to support the new encoding.

With the new encoding, verifying the first rule is straightforward; we can apply the chosen cut protocol to select a cell with an arrow in each region.

To verify the third rule, we first create a public adjacency matrix of all pairs of regions. Then, we select an entry in the matrix corresponding to the two cells containing the matched arrows to show that they are in non-adjacent regions.

\subsection{Encoding}
We encode an up arrow, right arrow, down arrow, and left arrow by $E_4(1)$, $E_4(2)$, $E_4(3)$, and $E_4(4)$, respectively. We encode an arrow that has already been verified by $E_4(5)$. Also, we encode an empty cell with no arrow by $E_4(0)$. Let $R_1,R_2,...,R_k$ be the $k$ regions of the Toichika grid. We encode a region $R_i$ by $E_k(i)$.

Each cell in region $R_i$ is encoded by a sequence of $k+4$ cards in the form $E_4(j) \circ E_k(i)$, where $j$ corresponds to the direction of an arrow in that cell (or an empty cell), and $\circ$ denotes concatenation of two sequences.

\subsection{Showing that a Number is in $\{1,2,3,4\}$} \label{nz}
Suppose $P$ has a face-down sequence $E_4(i)$. $P$ wants to show $V$ that $1 \leq i \leq 4$, i.e. $i \neq 0$ and $i \neq 5$, without revealing the value of $i$. $P$ can do so by applying the uniqueness verification protocol in Section \ref{unique} to show that the sequence consists of three \mybox{0}s and one \mybox{1}. Note that this protocol preserves the original state of the sequence.

\subsection{Showing that Two Numbers are Equal} \label{eq}
Suppose $P$ has two face-down sequences $E_4(i)$ and $E_4(j)$ ($1 \leq i,j \leq 4$). $P$ wants to show $V$ that $i=j$ without revelaing the values of $i$ and $j$. $P$ can do so by performing the following steps.
\begin{enumerate}
	\item Construct a $2 \times 4$ matrix $M$ by placing $E_4(i)$ in Row 1 and $E_4(j)$ in Row 2.
	\item Apply the pile-shifting shuffle to $M$.
	\item Turn over all cards in $M$ to show that the \mybox{1}s in Row 1 and Row 2 are at the same column (otherwise $V$ rejects).
\end{enumerate}

Note that this protocol does not preserve the original states of the sequences.

\subsection{Modified FirstNonZero Protocol}
The modified protocol works exactly the same way as the original protocol in Section \ref{nonzero}, except that we use multiple cards to encode each term instead of one.

Suppose $P$ wants to show $V$ that $x_\ell$ is the first nonzero term of a sequence $S=(x_1,x_2,...,x_n)$, where $0 \leq x_i \leq 5$ for every $i$, and also wants to show that $1 \leq x_\ell \leq 4$ (in particular, $x_\ell \neq 5$). Optionally, $P$ may also publicly edit the value of $x_i$.

Each $x_i$ is encoded by a sequence $a_i$ of $k+4$ face-down cards, where the first four cards are $E_4(x_i)$, and the last $k$ cards are arbitrary cards. $P$ performs the following steps.
\begin{enumerate}
	\item Publicly append sequences $b_1,b_2,...,b_{n-1}$, all of them being $E_{k+4}(0)$s, to the left of $A$. Call this new sequence $C=(c_1,c_2,...,c_{2n-1}) = (b_1,b_2,...,b_{n-1},$ $a_1,a_2,...,a_n)$.
	\item Turn over all face-up cards. Apply the chosen cut protocol to $C$ to choose the sequence $c_{\ell+n-1}=a_\ell$.
	\item Show that the first four cards of $c_{\ell+n-1}$ is in the form $E_4(i)$ for some $1 \leq i \leq 4$, using the protocol in Section \ref{nz}.
	\item As the chosen cut protocol preserves the cyclic order of $C$, turn over the first four cards of sequences $c_\ell,c_{\ell+1},...,c_{\ell+n-2}$ to show that they are all $E_4(0)$s (otherwise $V$ rejects). Now $V$ is convinced that $x_\ell$ is the first nonzero term of $S$, and that $1 \leq x_\ell \leq 4$, without knowing $\ell$.
	\item Optionally, $P$ may publicly replace the sequence $c_{\ell+n-1}$ with another desired sequence.
	\item Continue the chosen cut protocol to the end. Then, remove sequences $b_1,b_2,...,$ $b_{n-1}$. The sequence $A$ is now reverted back to its original state.
\end{enumerate}

\subsection{Non-Adjacent Regions Verification} \label{adj}
At first, $P$ publicly constructs a $k \times k$ matrix $B$ of cards such that
$$B(i,j)=\begin{cases}
\boxed{2}, &\text{ if $i=j$;}\\
\boxed{1}, &\text{ if $i \neq j$ and $R_i$ is adjacent to $R_j$;}\\
\boxed{0}, &\text{ otherwise,}
\end{cases}$$
where $B(i,j)$ denotes the element in the $i$-th row and $j$-th column of $B$.

Suppose $P$ has two face-down sequences $E_k(i)$ and $E_k(j)$. $P$ wants to show $V$ that $B(i,j)$ is a \mybox{0}, i.e. regions $R_i$ and $R_j$ are not adjacent, without revealing $i$ and $j$. $P$ can do so by performing the following steps.
\begin{enumerate}
	\item Apply the chosen cut protocol to select the $i$-th row of $B$ (by placing $E_k(i)$ in Step 1(b)).
	\item Apply the chosen cut protocol again to select the $j$-th card in that row (by placing $E_k(j)$ in Step 1(b)).
	\item Turn over the selected card to show that it is a $\mybox{0}$ (otherwise $V$ rejects).
	\item Continue the chosen cut protocol to the end. $B$ is now reverted back to its original state.
\end{enumerate}

\subsection{Setup}
On every cell in each region $R_i$ of the Toichika grid, $P$ publicly places a sequence $E_k(i)$. For each cell already containing an arrow, $P$ publicly appends a sequence $E_4(j)$ (depending on the direction of the arrow) to the left of the sequence on the cell. For each initially empty cell, $P$ secretly appends a face-down sequence $E_4(j)$ (depending on the direction of the arrow or an empty cell, according to $P$'s solution) to the left of the sequence on the cell.

$P$ extends the grid by $n-1$ cells in all directions by publicly placing ``dummy sequences'' $E_{k+4}(0)$s aound the grid. Then, turn over all face-up cards. We now have a $(3n-2) \times (3n-2)$ matrix of cards.

$P$ also publicly constructs a $k \times k$ adjacenecy matrix $B$ as defined in Section \ref{adj}.

\subsection{Main Protocol}
To verify the first rule, $P$ performs the following steps for $k$ iterations, one for each region.
\begin{enumerate}
	\item Pick up the sequences on all cells in a region.
	\item Apply the chosen cut protocol to select a sequence on a cell with an arrow.
	\item Show that the first four cards of that sequence is in the form $E_4(i)$ for some $1 \leq i \leq 4$, using the protocol in Section \ref{nz}.
	\item Turn over the first four cards of all other sequences in that region to show that they are all $E_4(0)$s (otherwise $V$ rejects).
	\item Place all cards back to the grid.
\end{enumerate}

To verify the second and third rules, $P$ performs the following steps for $k/2$ iterations, one for each pair of matched arrows.
\begin{enumerate}
	\item Apply the chosen cut protocol to select a sequence on a cell with one of the two arrows. Call this sequence $c_1$.
	\item Show that the first four cards of $c_1$ is in the form $E_4(i)$ for some $1 \leq i \leq 4$, using the protocol in Section \ref{nz}.
	\item Take the $n-1$ cards on a path from $c_1$ in the direction to the north. Let $A_1=(a_{(1,1)},a_{(1,2)},...,a_{(1,n-1)})$ be the sequence of these cards in this order from the nearest to the farthest. Analogously, let $A_2, A_3$, and $A_4$ be the sequences of the $n-1$ cards on a path from $c_1$ in the direction to the east, south, and west, respectively. Stack each sequence into a single stack.
	\item Apply the chosen cut protocol to select a stack $A_i$, where $E_4(i)$ is the first four cards of $c_1$ (by placing $E_4(i)$ in Step 1(b)).
	\item Apply the modified FirstNonZero protocol to the sequence $A_i$ to select a sequence corresponding to the first arrow in that direction. Call this sequence $c_2$.
	\item Swap the first and third cards of $c_2$. Also, swap the second and fourth cards of $c_2$. Note that this reverses the direction of the arrow on the cell containing $c_2$.
	\item Show that the first four cards of $c_1$ and $c_2$ now encode the same number, i.e. the two arrows now point to the same direction, using the protocol in Section \ref{eq}.
	\item Replace the first four cards of $c_1$ and $c_2$ with $E_4(5)$s.
	\item Take the last $k$ cards of $c_1$ and $c_2$. Suppose they are $E_k(x)$ and $E_k(y)$, respectively.
	\item Apply the protocol in Section \ref{adj} to show that $B(x,y)$ is a \mybox{0}, i.e. $R_x$ and $R_y$ are not adjacent.
	\item Place all cards back to the grid.
\end{enumerate}

If the verification passes, then $V$ accepts.

Our ZKP protocol for Toichika uses $\Theta(kn^2)$ cards and $\Theta(k)$ shuffles.

\subsection{Proof of Correctness and Security}
We will prove the perfect completeness, perfect soundness, and zero-knowledge properties of our protocol for Toichika.

\begin{lemma}[Perfect Completeness] \label{lem4}
If $P$ knows a solution of the Toichika puzzle, then $V$ always accepts.
\end{lemma}

\begin{proof}
Suppose $P$ knows a solution of the puzzle.

Since each region contains exactly one arrow, the verification of the first rule will pass for every region.

Now consider the verification of the second and third rules of each pair of matched arrows.

Since both arrows have never been verified before, the first four cards on the cell containing each of them will be in the form $E_4(i)$ for some $1 \leq i \leq 4$, so Step 2 will pass.

Since there is no other arrow between the two arrows, the selected sequence in Step 5 will be on the cell containing the other arrow, and its first four cards must be the form $E_4(j)$ for some $1 \leq j \leq 4$. Therefore, Step 5 will pass. Also, as the two arrows point to opposite directions, after Step 6 they will point to the same direction, so Step 7 will pass.

Moreover, since the two arrows are in non-adjacent regions, Step 10 will pass.

Hence, we can conclude that $V$ always accepts.
\end{proof}

\begin{lemma}[Perfect Soundness] \label{lem5}
If $P$ does not know a solution of the Toichika puzzle, then $V$ always rejects.
\end{lemma}

\begin{proof}
We will prove the contrapositive of the statement. Suppose that $V$ accepts, which means the verification passes for every iteration.

As the verification of the first rule pass for every region, the first rule of the puzzle is satisfied.

Now consider each iteration of the verification of the second and third rules.

Since Step 2 passes, the first four cards of $c_1$ must be in the form $E_4(i)$ for some $1 \leq i \leq 4$, meaning that an arrow on that cell has never been verified before.

Since Step 5 passes, the first four cards of $c_2$ must be in the form $E_4(j)$ for some $1 \leq j \leq 4$, meaning that an arrow on that cell has never been verified before. Moreover, there is no arrow between the cells containing $c_1$ and $c_2$.

In Step 6, the direction of the arrow on the cell containing $c_2$ is reversed. Since Step 7 passes, the two arrows originally point towards each other, satisfying the second rule of the puzzle.

Since Step 9 passes, the two arrows are in non-adjacent regions, satisfying the third rule of the puzzle.

As this is true for all $k/2$ iterations, all $k$ arrows are matched accoring to the second and third rules of the puzzle. Hence, we can conclude that $P$ must know a valid solution of the puzzle.
\end{proof}

\begin{lemma}[Zero-Knowledge] \label{lem6}
During the verification, $V$ obtains no information about $P$'s solution.
\end{lemma}

\begin{proof}
It is sufficient to show that all distributions of cards that are turned face-up can be simulated by a simulator $S$ that does not know $P$'s solution. Note that the zero-knowledge property of the chosen cut protocol and the FirstNonZero protocol has been proved in the proof of Lemma \ref{lem3}.
\begin{itemize}
	\item In Steps 3 and 5 of the uniqueness verification protocol in Section \ref{unique}, due to the uniformly random permutation, the sequence has an equal probability to be any of the $q!$ permutations. Hence, these steps can be simulated by $S$.
	\item In Step 3 of the protocol in Section \ref{eq}, due to the pile-shifting shuffle, the two \mybox{1}s has an equal probability to be at any of the four columns. Hence, these steps can be simulated by $S$.
	\item In Step 3 of the protocol in Section \ref{adj}, the turned over card is a \mybox{0}, so it can be simulated by $S$.
	\item In the main protocol, there is only one deterministic pattern of the cards that are turned face-up, so the whole protocol can be simulated by $S$.
\end{itemize}

Therefore, we can conclude that $V$ obtains no information about $P$'s solution.
\end{proof}

\section{Future Work}
We proposed a new physical card-based protocol to verify the first nonzero term of a sequence, and applied it to develop ZKP protocols for three logic puzzles: ABC End View, Goishi Hiroi, and Toichika. A possible future work includes exploring methods to physically verify other sequence-related functions, and developing ZKP protocols for other logic puzzles.

\subsubsection*{Acknowledgement}
The author would like to thank Daiki Miyahara and Kyosuke Hatsugai for a valuable discussion on this research.

\end{document}